\documentclass[runningheads]{llncs}
\usepackage{graphicx,color,amssymb}
\usepackage{hyperref}

\usepackage{amsmath}
\usepackage{amssymb}
\usepackage{bbold}
\usepackage{centernot}
\usepackage{listings}
\usepackage{varwidth}
\usepackage{float}
\usepackage{wrapfig}
\usepackage[nocompress]{cite}

\newcommand\bdp{\textsc{BD-threshold}}

\newfloat{program}{thp}{lop}
\floatname{program}{Program}

\newenvironment{proof*}[1]
  {\renewcommand{\proofname}{#1}%
   \begin{proof}}
  {\end{proof}}

\newcommand{\alphabet}{\Sigma}
\newcommand{\setofmeasurables}{\mathcal F}
\newcommand\ind[1]{\mathbb{1}_{#1}}
\newcommand{\dm}{\ind{\ne}}
\newcommand\kanta{K_{\alpha}}

\newcommand\abra[1]{\langle #1 \rangle}
\newcommand\lab{\ell}
\newcommand\dist[1]{\mathit{Dist}(#1)}

\newcommand\mcal[1]{\mathcal{#1}}
\newcommand\lmc{\mathcal{M}}
\newcommand\Q{{\mathbb Q}}
\newcommand\cyl{{C}}

\newcommand\bd[1]{\mathit{bd}_\alpha}
\newcommand\bdd[1]{\mathit{bd}_{#1}}

\newcommand\cutout[1]{}

\newcommand\pathfin[2]{\mathit{Paths}_{#2}(#1)}
\newcommand\tv[1]{\mathit{tv}_{#1}}
\newcommand\tva{\tv{\alpha}}
\newcommand\ga{\Gamma_\alpha}

\newcommand\Reals{\mathbb R_{\ge 0}}

\usepackage[bold,full]{complexity}

\usepackage{listings}

\spnewtheorem{fact}{Fact}{\bfseries}{\itshape}

\begin{document}

\title{Bisimilarity Distances~for~Approximate Differential Privacy}%
\author{Dmitry Chistikov\inst{1}%
\and
Andrzej S. Murawski\inst{2}%
\and
David Purser\inst{1}%
}

\authorrunning{D. Chistikov et al.}
\institute{Centre for Discrete Mathematics and its Applications (DIMAP) \& \\Department of Computer Science, University of Warwick, UK \\
\and
Department of Computer Science, University of Oxford, UK \\
}

\maketitle              %
\begin{abstract}
Differential privacy is a widely studied notion of privacy for various models of computation.
Technically, it is based on measuring differences between probability distributions.
We study $\epsilon,\delta$-differential privacy in the setting of labelled Markov chains.
While the exact differences relevant to $\epsilon,\delta$-differential privacy are not computable
in this framework, we propose a computable bisimilarity distance that yields a sound technique for measuring $\delta$,
the parameter that quantifies deviation from pure differential privacy.
We show this bisimilarity distance is always rational, the associated threshold problem is in $\NP$,
and the distance can be computed exactly with polynomially many calls to an $\NP$ oracle.

\end{abstract}

\section{Introduction}

Bisimilarity distances were introduced by \cite{desharnais2002metric,desharnais2004metrics}, as a metric analogue of classic probabilistic bisimulation~\cite{larsen1991bisimulation},
to overcome the problem that bisimilarity is too sensitive to minor changes in probabilities. 
Such robustness is highly desirable, because probabilistic automata arising in practice may often be based on approximate probability values, extracted or learnt from  real world data.

In this paper, we study the computation of bisimilarity distances related to differential privacy.
Differential privacy~\cite{dwork2006calibrating} is a security property that ensures that a small perturbation of the input leads to only a small perturbation in the output, so that observing the output makes it difficult to determine whether
a particular piece of information was present in the input.
A variant, $\epsilon$-differential privacy, considers the ratio difference (rather than the absolute difference) between probabilities.

We will be concerned with the more general concept of  $\epsilon,\delta$-differential privacy, also referred to as \emph{approximate differential privacy}.
The $\delta$ parameter allows one to assess to what degree $\epsilon$-differential privacy {(``pure differential privacy''}) was achieved. 
We will design a version of bisimilarity distance which will constitute a sound upper bound on~$\delta$, thus providing a reliable measure of security.

From a verification perspective, a natural question is how to analyse systems with respect to $\epsilon, \delta$-differential privacy.
We carry out our investigations in the setting where the systems are \emph{labelled Markov chains (LMC)},
abstractions of autonomous systems with probabilistic behaviour and under partial observability.
States of an LMC~$\lmc$ can be thought of as generating probability distributions on sets of traces, and these sets are taken to correspond to observable events.
Let $\lmc$ be a system, and suppose $s$ and $s'$ are two states (configurations) of $\lmc$. Then we will say that $s$ and $s'$ satisfy $\epsilon, \delta$-differential privacy if the distributions on traces from these states are sufficiently close.
We consider the following problem: given an LMC $\lmc$, states $s$ and $s'$, and a value of $\epsilon$, determine $\delta$ such that $s$ and $s'$ satisfy $\epsilon, \delta$-differential privacy.
Unfortunately,  the smallest of such $\delta$ is not computable~\cite{Kie18}, which motivates our search for upper bounds.

In the spirit of generalised bisimilarity pseudometrics~\cite{chatzikokolakis2014generalized}, our distance, denoted $\bd{\alpha}$, is based on the Kantorovich-style lifting of distance between states to distance between distributions. 
However, because the underpinning distances in our case turn out not to be metrics, the setting does not quite fit into the standard picture, which presents a technical challenge.
We discuss how the proposed distance may be computed, using techniques from linear programming, linear real arithmetic, and computational logic. 
Our first result is that the distance always takes on rational values of polynomial size with respect to the size of the LMC and the bit size of the probability values associated with transitions (Theorem~\ref{thm:polyrational}).

This is then used to show that the associated threshold problem (``is $\bd{\alpha}$ upper-bounded by a given threshold value for two given states?'') is in $\NP$ (Theorem~\ref{thm:bdp}).
Note that the distance can be approximated to arbitrary precision by solving polynomially many instances of the threshold problem. Finally, we show that the distance can be computed exactly in polynomial time, given an $\NP$ oracle (Theorem~\ref{thm:main}).
This places it in (the search version of) $\NP$, leaving the possibility of polynomial-time computation open.

\paragraph{\it\bf Related Work} Chatzikokolakis et al.~\cite{chatzikokolakis2014generalized} have advocated the development of Kantorovich pseudometrics,
instantiated with any metric distance function (rather than absolute value) in the context of differential privacy.
They did not discuss the complexity of calculating such pseudometrics, but asked whether it was possible  to extend their techniques to $\epsilon,\delta$-differential privacy.
Our paper shows the extent to which this can be achieved; the technical obstacle that we face is that our distances are not metrics. To the best of our knowledge, no complexity results on differential privacy for Markov chains have previously appeared in the literature, and we are the first to address this gap.

The computation of the standard bisimilarity distances has been the topic of a long running line of research~\cite{van2017probabilistic}, starting with approximation~\cite{van2007approximating}.
The distance was eventually determined to be computable in polynomial time using the ellipsoid method to solve an implicit linear program of exponential size \cite{chen2012complexity}. This technique turns out slow in practice and further techniques have been developed which are faster but do not have such strong complexity guarantees \cite{bacci2013fly,tang2016computing}.
Because of the two-sided nature of our distances, the main system of constraints that we introduce in our work involves a maximum of two quantities. This nonlinearity at the core of the problem prevents us from relying on the ellipsoid method and explains the gap between our $\NP$ upper bound and the polynomial-time algorithms of~\cite{chen2012complexity}.

Tschantz et al.~\cite{tschantz2011formal} first studied differential privacy using a notion similar to bisimulation, which was extended to a more general class of bisimulation relations by Xu et al.~\cite{xu2014metrics}.
{Both consider only $\epsilon$-differential privacy, i.e. ratio differences, but do not examine how these could be computed.}

An alternative line of research by Barthe et al.~\cite{barthe2012probabilistic} concerns formal mechanised proofs of differential privacy.
Recently, that direction has been related to coupling proofs~\cite{barthe2015relational} -- this still requires substantial effort to choose the coupling, although recent techniques 
have improved this~\cite{albarghouthi2017synthesizing}. We complement this line of research by taking an algorithmic verification-centred approach.

The remainder of the paper is arranged as follows. Section~\ref{sec:lmc} introduces the basic setting of labelled Markov chains.
In section~\ref{sec:dp}, we discuss $\epsilon,\delta$-differential privacy and in section~\ref{sec:bd} we define our distance.
Section~\ref{sec:kant} develops technical results on our extended case of Kantorovich lifting.
These are subsequently used in section~\ref{sec:comp} to underpin techniques for computing the relevant distances.

\section{Labelled Markov Chains\label{sec:lmc}}

Given a finite set $S$, let $\dist{S}$ be the set of probability distributions on $S$.
\begin{definition}
A \emph{labelled Markov chain} (LMC) $\lmc$ is a tuple $\abra{S,\alphabet,\mu,\lab}$,
where 
$S$ is a finite set of states, 
$\alphabet$ is a finite alphabet, 
$\mu: S\to \dist{S}$ is the transition function 
and $\lab:S\to\alphabet$ is the labelling function.
\end{definition}
Like in~\cite{van2017probabilistic,chen2012complexity,bacci2013fly,tang2016computing}, our definition features labelled states. Variations, such as transition labels, can be easily accommodated within the setting.
We also assume that all transition probabilities are rational, represented as a pair of binary integers. The bit sizes of these integers form part of the bit size of the representation $|\mcal{M}|$. We will often write $\mu_s$ for $\mu(s)$.

In what follows, we  study probabilities associated with   infinite sequences of labels 
generated by LMC's.
We specify the relevant probability spaces next using standard measure theory~\cite{bill86,baier2008principles}.
Let us start with the definition of cylinder sets.
\begin{definition}
A subset $\cyl\subseteq\alphabet^\omega$ is a \emph{cylinder set} if there exists $u \in \alphabet^\ast$ such that
$\cyl$ consists of all infinite sequences from $\alphabet^\omega$ whose prefix is $u$. We then write $\cyl_u$
to refer to $C$.
\end{definition}

Cylinder sets play a prominent role in measure theory in that their finite unions can be used as  a generating family (an algebra) 
for the set $\mcal{F}$ of  measurable subsets of $\alphabet^\omega$ (the cylindric $\sigma$-algebra).
What will be important for us is that any measure $\nu$ on $\mcal{F}$ is uniquely determined by its values on cylinder sets.
Next we show how to assign a measure $\nu_s$ on $\mcal{F}$ to an arbitrary state of an LMC.
We start with several auxiliary definitions.
\begin{definition}
Given $\lmc=\abra{S,\alphabet,\mu,\lab}$, let $\mu^+:S^+\to[0,1]$ and $\lab^+:S^+\to\alphabet{}^+$ be the natural 
extensions of $\mu$ and $\lab$ to $S^+$, i.e.
$\mu^+(s_0\cdots s_k) = \prod_{i=0}^{k-1} \mu(s_i)(s_{i+1})$ and $\lab^+(s_0\cdots s_k) = \lab(s_0)\cdots \lab(s_k)$, 
where $k\ge 0$ and $s_i\in S$ ($0\le i\le k$).
Note that, for any $s\in S$, we have $\mu^+(s)=1$.
Given $s\in S$, let $\pathfin{\lmc}{s}$ be the subset of $S^+$ consisting of all sequences that start with $s$.
\end{definition}

\begin{definition}
Let $\lmc=\abra{S,\alphabet,\mu,\lab}$ and $s\in S$. We define $\nu_s: \mcal{F} \to [0,1]$ to be the unique measure on $\mcal{F}$
such that for any cylinder $\cyl_u$ we have
\[
\nu_s(\cyl_u) = \sum\{ \,\mu^+(p)\,\,|\,\, p\in \pathfin{\lmc}{s},\, \lab^+(p) = u\, \}.
\]
\end{definition}
Our aim will be to compare states of labelled Markov chains from the point of view of differential privacy.
Note that two states $s,s'$ can be viewed as  indistinguishable if $\nu_{s}=\nu_{s'}$.
If they are not indistinguishable then the difference between them can be quantified using the \emph{total variation distance}, defined by
$\tv{}(\nu,\nu') = \sup_{E\in \mcal{F}}  |\nu(E)-\nu'(E)|$. 
Given $\lmc=\abra{S,\alphabet,\mu,\lab}$ and $s,s'\in S$, we shall write $\tv{}(s,s')$ to refer to $\tv{}(\nu_s,\nu_{s'})$. 
\begin{remark}\label{rem:bad}
$\tv{}(s,s')$ turns out  surprisingly difficult to compute: it is undecidable whether the distance is strictly greater than a given threshold,
and the non-strict  variant of the problem (``greater or equal'') is not known to be decidable~\cite{Kie18}.
\end{remark}
To measure probabilities relevant to differential privacy, we will need to study a more general variant $\tv{\alpha}$ of the above distance, which we introduce next.

\section{Differential Privacy\label{sec:dp}}
Differential privacy is a mathematical guarantee of privacy due to Dwork et al~\cite{dwork2006calibrating}. 
It is a property similar to non-interference: the aim is to ensure that inputs which are related in some sense lead to very similar outputs. 
The notion requires that for two related states there only ever be a small change in output probabilities, and therefore discerning the two is difficult, which maintains the privacy of the states. 
Below we cast the definition in the setting of labelled Markov chains.
\begin{definition}\label{def:dp}
Let $\lmc=\abra{S,\alphabet,\mu,\lab}$ be a labelled Markov chain and let $R\subseteq S\times S$ be a symmetric relation.
Given $\epsilon\ge 0$ and $\delta\in [0,1]$, we say that
$\lmc$ is $\epsilon, \delta$-differentially private (wrt $R$) if, for any $s,s'\in S$ such that $(s,s')\in R$, we have
\[
\nu_s(E) \leq e^{\epsilon}\cdot \nu_{s'}(E) + \delta
\]
for any measurable set $E\in\mcal{F}$.
\end{definition}
\begin{remark}
Note that each state $s\in S$ can be viewed as defining a random variable $X_s$ with outcomes from $\alphabet^\omega$ such that $P(X_s\in E) = \nu_s(E)$. Then the above can be rewritten as $P(X_s \in E) \le e^\epsilon\, P(X_{s'} \in E) + \delta$, 
which matches the definition from~\cite{dwork2006calibrating}, where one would consider $X_s, X_{s'}$ neighbouring in some natural sense. 
\end{remark}
The above formulation is often called \emph{approximate differential privacy}. For $\delta=0$, one talks about (pure) \emph{$\epsilon$-differential privacy}.
Note that then the above definition boils down to measuring the ratio between the probabilities of possible outcomes.
$\delta$ is thus an indicator of the extent to which $\epsilon$-differential privacy holds for the given states.
Intuitively, one could interpret $\epsilon,\delta$-differential privacy as ``$\epsilon$-differential privacy with probability at least $1-\delta$"~\cite{V17}.
Our work is geared towards obtaining sound upper bounds on the value of $\delta$ for a given $\epsilon$.

\begin{remark}
What it means for two states to be related (as specified by $R$) is to a large extent domain-specific. In general, $R$ makes it possible to spell out which
states should not appear too different and, consequently, should enjoy a quantitative amount of privacy.
In the typical database scenario, one would relate database states that differ by just one person.
In our case, we refer to states of a machine, 
for which we would like it to be indiscernible as to which was the start state (we assume the states are hidden and the traces are observable).
\end{remark}
To rephrase the inequality underpinning differential privacy in a more succinct form, it will be convenient to work with the \emph{skewed distance} $\Delta_\alpha$,
first introduced by Barthe et al~\cite{barthe2012probabilistic} in the context of Hoare logics and $\epsilon,\delta$-differential privacy.
\begin{definition}[Skewed Distance]
\label{def:skd}
For $\alpha \geq 1$, let $\Delta_\alpha: \Reals\times\Reals \to\Reals$ be defined by 
$\Delta_\alpha(x,y) = \max\{x - \alpha y, \,y - \alpha x, \, 0\}$.
\end{definition}
\begin{remark}
It is easy to see that $\Delta_\alpha$ is anti-monotone with respect to $\alpha$. In particular, because $\alpha\geq 1$, we have $\Delta_\alpha(x,y) \le \Delta_1(x,y)= |x-y|$.
Observe that $\Delta_2(9,3)=9-2\times 3 = 3$, $\Delta_2(9,6) = 0$ and $\Delta_2(6,3)=0$.
Note that $\Delta_2(x,y)=0$ need not imply $x=y$, i.e. $\Delta_2$ is not a metric.
Note also that the triangle inequality may fail: $\Delta_2(9,3) > \Delta_2(9,6)+\Delta_2(6,3)$, i.e. $\Delta_2$ is not a pseudometric\footnote{A pseudometric must satisfy $m(x,x)=0$, $m(x,y)=m(y,x)$ and $m(x,z)\le m(x,y)+m(y,z)$. 
For metrics, one additionally requires that $m(x,y)=0$ should imply $x=y$.}.
This will complicate our technical development, because we will not be able to use the framework of~\cite{chatzikokolakis2014generalized} directly.
\end{remark}
The significance of the skewed distance will be seen shortly in Fact~\ref{f:defeqdp}. We first introduce the skewed analogue of the total variation distance called $\tv{\alpha}$, for which $\tv{}$ is a special case ($\alpha=1$).
\begin{definition}
\label{def:tva}
Let $\alpha\ge 1$. Given two measures $\nu,\nu'$ on $(\alphabet^\omega, \mcal{F})$, let
\[
\tv{\alpha}(\nu,\nu') = \sup_{E \in \mcal{F}} \Delta_\alpha(\nu(E),\nu'(E)).
\]
\end{definition}
Following the convention for $\tv{}$, $\tv{\alpha}(s,s')$ will stand for $\tv{\alpha}(\nu_s,\nu_{s'})$.
Fact~\ref{f:defeqdp} is an immediate corollary of Definitions~\ref{def:dp}, \ref{def:skd}, and~\ref{def:tva}.
\begin{fact}
\label{f:defeqdp}
$\mcal{M}$ is $\epsilon,\delta$-differentially private wrt $R$ if and only if, for all $s,s'\in S$ such that $(s,s')\in R$, we have $\tv{\alpha}(s,s') \leq \delta$, where $\alpha = e^\epsilon$.
\end{fact}
Some values of $\tv{\alpha}$ are readily known. For instance, the distance between any bisimilar states turns out to be zero.
\begin{definition}
A probabilistic bisimulation on an LMC $\lmc=\abra{S,\alphabet,\mu,\lab}$  is 
an equivalence relation $R\subseteq S\times S$ such that  if $(s,s') \in R$ then $\lab(s)=\lab(s')$ and for all $X \in S/R$, $\sum_{u\in X}\mu(s)(u) =\sum_{u\in X} \mu(s')(u)$,
i.e. related states have the same label and probability of transitioning into any given  equivalence class.
\end{definition}
It is known that probabilistic bisimulations are closed under union and hence there exists a largest one, written $\sim$ and called \emph{probabilistic bisimilarity}.
Two states are called \emph{bisimilar}, written $s\sim s'$, if $(s,s')\in \sim$. Equivalently, this means that the pair $(s,s')$ belongs to a probabilistic bisimulation.
It follows from~\cite[Proposition 9, Lemma 10]{chen2012complexity}, that for bisimilar $s,s'$, we have $\tv{1}(s,s')=0$. As $\tva(s,s')\le \tv{1}(s,s')$ we obtain the following.
\begin{lemma}\label{lem:bis}
If $s\sim s'$ then $\tva(s,{s'}) =0$.
\end{lemma}
In contrast to~\cite{chatzikokolakis2014generalized}, the converse will not hold.
\begin{example}\label{ex:kernel}
In the LMC shown in Figure~\ref{fig:eg:nonbisim}, states $s_0$ and $s_1$ are \emph{not} bisimilar.
To see this, observe first that  $s_2$ must be the only state in its equivalence class with respect $\sim$, because other states have different labels.
Now note that the probabilities of reaching $s_2$ from $s_0$ and $s_1$ respectively are different ($0.4$ vs $0.6$).

However,  for $\alpha=1.5$,  we have $\tv{\alpha}(s_0,s_1)=0$, because $\Delta_{\alpha}(0.6,0.4)= \max(0.6-1.5\cdot 0.4, 0.4-1.5\cdot 0.6, 0)=0$.
\begin{figure}[t]
\centering
\includegraphics[scale=0.4]{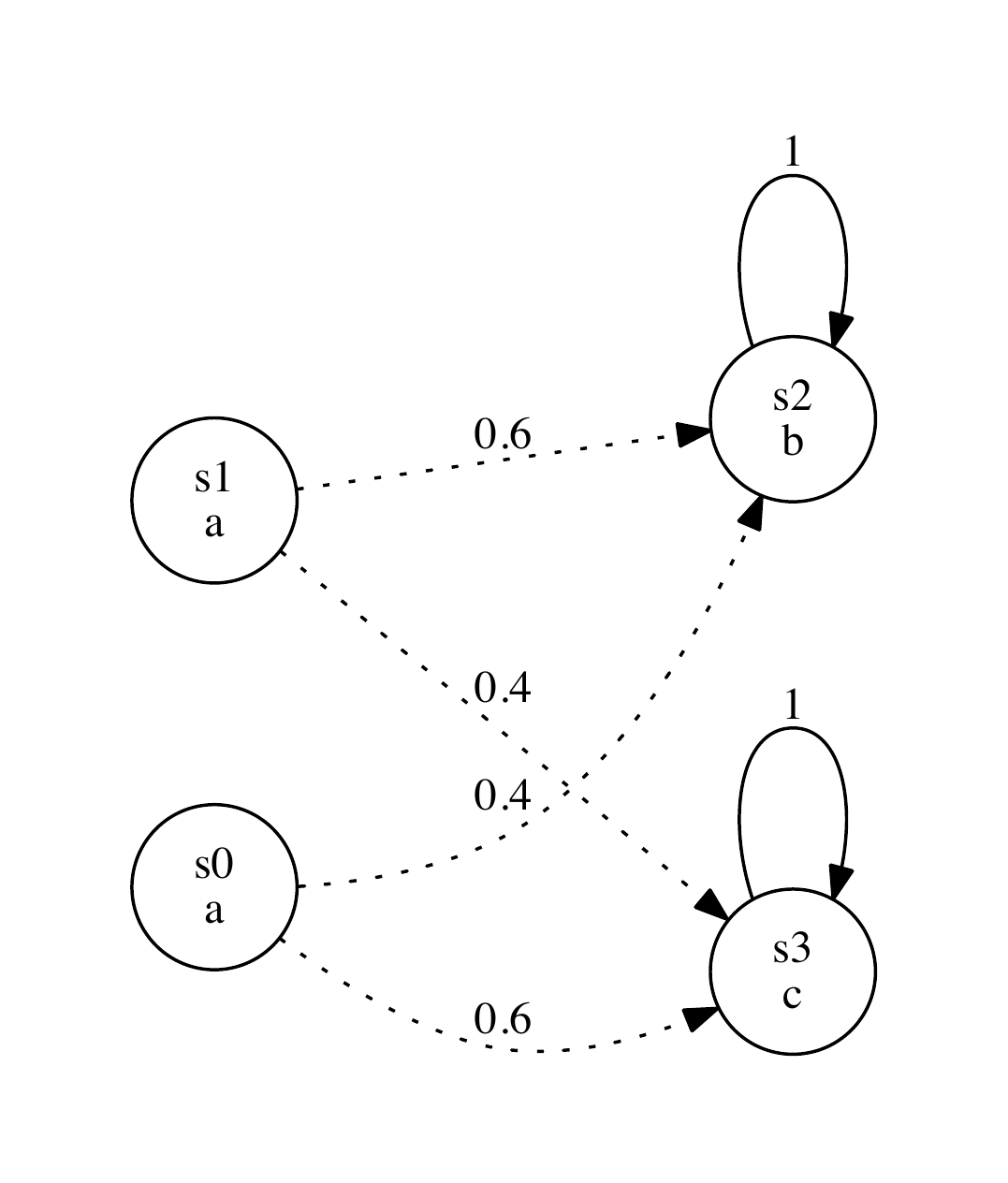}
\caption{States 1 and 2 are not bisimilar, but $\tv{1.5}(s_0,s_1) = 0$.}
\label{fig:eg:nonbisim}
\end{figure}
\end{example}
In an ``acyclic'' system, $\tv{\alpha}$ can be calculated by exhaustive search: the natural algorithm 
is doubly exponential, as one needs to consider all possible events over all possible traces.
However, in general, $\tv{\alpha}$ is not computable~(Remark~\ref{rem:bad}).
Thus, in the remainder of the paper, we shall introduce and study another distance $\bd{}$.
It will turn out possible to compute it and it will provide a sound method
for bounding $\delta$ for $\ln(\alpha),\delta$-differential privacy.  
Our main result will be Theorem~\ref{thm:main}: the new distance can be calculated in polynomial time, assuming an $\NP$ oracle. Pragmatically, this means that this new distance can be computed efficiently, assuming access to an appropriate satisfiability or theory solver.

\section{Skewed Bisimilarity Distance\label{sec:bd}}

Our distance will be defined in the spirit of bisimilarity distances~\cite{desharnais2002metric,desharnais2004metrics,chen2012complexity,chatzikokolakis2014generalized} through a fixed point definition
based on a variation of the Kantorovich lifting.
To motivate its shape, let us discuss how one would go about calculating $\tv{\alpha}$ recursively.
If $\lab(s)\neq \lab(s')$ then $\nu_s(C_{\lab(s)}) = 1$, $\nu_{s'}(C_{\lab(s)}) = 0$, therefore $\tva(s,s')=1$.
So, let us assume $\lab(s)=\lab(s')$. Given $E\subseteq\alphabet^\omega$ and $a\in \alphabet$, 
let $E_a = \{ w \in\alphabet^\omega\,|\, a w \in E\}$.  Then we have:
\begin{align*}
\tv{\alpha}(\nu_s,\nu_{s'}) &= \sup_{E \in \mcal{F}} \Delta_\alpha(\nu_s(E),\nu_{s'}(E))
\\ &= \sup_{E_{\lab{(s)}}\in \mcal{F}} \Delta_\alpha\big(\sum_{u\in S}\, \mu_s(u)\, \nu_u(E_{\lab(s)}),\sum_{u\in S}\mu_{s'}(u)\,\nu_u(E_{\lab(s)})\big).
\end{align*}
If we define $f: S\to [0,1]$ by $f(u) = \nu_u(E_{\lab(s)})$, this can be rewritten as
\[
\sup_{E_{\lab{(s)}}\in \mcal{F}} \Delta_\alpha\big(\sum_{u\in S}\, \mu_s(u)\, f(u),\sum_{u\in S}\mu_{s'}(u)\,f(u)\big).
\]
We have little knowledge of $f$, otherwise we could compute $\tv{\alpha}$, but  from the definition of $\tv{\alpha}$, we do know that $\Delta_\alpha(f(v),f(v')) \le \tv{\alpha}(v,v')$ for any $v,v'\in S$.
Consequently, the following inequality holds.
\begin{align*}
\tva(s,s') \le \sup_{\substack{f: S \to [0,1] \\ \forall v,v' \in S \Delta_\alpha(f(v), f(v')) \le \tv{\alpha}(v,v')}} \mkern-18mu\Delta_\alpha\big(\sum_{u\in S} \mu_s(u) f(u),\sum_{u\in S}\mu_{s'}(u)f(u)\big)
\end{align*}

The expression on the right is an instance of the Kantorovich lifting~\cite{Kan42,deng2009kantorovich}, which uses (``lifts'') the distance $\tv{\alpha}$ between states $s,s'$ to
define a distance between the distributions $\mu_s,\mu_{s'}$ associated with the states. We recall the definition of the Kantorovich distance between distributions 
in the discrete case, noting that then, for $\mu\in\dist{S}$,  we have $\int f d \mu = \sum_{u\in S} f(u) \,\mu(u)$.
\begin{definition}[Kantorovich] \label{def:kan}
Given $\mu,\mu' \in \dist{S}$ and  a pseudometric $m:S\times S\to[0,1]$, the \emph{Kantorovich distance} between $\mu$ and $\mu'$ is defined to be
$$ K(m)(\mu,\mu') = \sup_{\substack{f:S \to [0,1] \\ \forall v,v' \in S |f(v) - f(v')| \le m(v,v') }} \mkern-18mu\big|\int f d\mu - \int f d\mu'\big|.$$
\end{definition}
\begin{remark}
The Kantorovich distance is
 also known under other names (e.g. Hutchinson, Wasserstein distance), having been rediscovered several times in history~\cite{deng2009kantorovich}.
Chatzikokolakis et al. \cite{chatzikokolakis2014generalized} studied the Kantorovich distance and related bisimulation distances when the absolute value distance above is replaced with another metric.
For our purposes, instead of $|...|$, we need to consider $\Delta_\alpha$, even though $\Delta_\alpha$ is not a metric and $m$ may not be a pseudometric.
\begin{definition}[Skewed Kantorovich]\label{defn:sk}
Given $\mu,\mu' \in \dist{S}$ and a symmetric distance $d:S\times S\to[0,1]$, the \emph{skewed Kantorovich distance} between $\mu$ and $\mu'$ is defined to be
\[
\kanta(d)(\mu,\mu') = \sup_{\substack{f:S \to [0,1] \\ \forall v,v' \in S \ \Delta_\alpha(f(v), f(v')) \le d(v,v') }}\mkern-18mu\Delta_\alpha\big(\int f d\mu,\int f d\mu'\big)
\]
\end{definition}
\end{remark}
Note that setting $\alpha = 1$ gives the standard Kantorovich distance (Definition~\ref{def:kan}). Below we define a function operator, which will be used to define our distance.
\begin{definition}
Let $\ga: [0,1]^{S\times S}\to [0,1]^{S\times S}$ be defined as follows.
$$\Gamma_\alpha(d)(s,s') = \begin{cases} 
K_\alpha(d)(\mu_s, \mu_{s'}) & \lab(s) = \lab(t) \\ 
1 & \lab(s) \ne \lab(t)\end{cases}$$
\end{definition}
Note that $[0,1]^{S\times S}$ equipped  with the pointwise order, written $\sqsubseteq$, 
is a complete lattice and that $\Gamma_\alpha$ is monotone with respect that order (larger $d$ permit more functions, thus larger supremum).
Consequently, $\ga$ has a least fixed point~\cite{tarski1955lattice}. 
We take our distance to be exactly that point.
\begin{definition}[Skewed Bisimilarity Distance]
Let $\bd{\alpha}: S\times S\to [0,1]$ be the least fixed point  of $\ga$.
\end{definition}
\begin{remark}\label{rem:ufp}
Recall that the least fixed point is equal 
to the least pre-fixed point ($\min \{ d\,|\, \ga(d)\sqsubseteq d\}$).
\end{remark}
Recall our initial remarks about the Kantorovich distance $\kanta(\tva)(\mu_s,\mu_{s'})$ overapproximating $\tva(s,s')$.
They can be summarised by $\tva \sqsubseteq \kanta(\tva)$, i.e. $\tva$ is a post-fixed point of $\kanta$.
Since we want to bound $\tv{\alpha}$ as closely as possible, we can show that the least fixed point $\bd{\alpha}$ also bounds $\tva$ from above.
\begin{lemma}\label{lem:bounds}
$\tva \sqsubseteq \bd{\alpha}$.
\end{lemma}
\begin{remark}
The lemma is an analogue of Theorem 2~\cite{chatzikokolakis2014generalized}.
Its proof in~\cite{xu2015formal} relied on the fact that the counterpart of $\Delta_\alpha$ 
was a metric, which is not true in our case (unless $\alpha=1$). 
\end{remark}

Just like $\Delta_\alpha$ is anti-monotone with respect to $\alpha$, so is $\bd{\alpha}$.
This means that $\bd{\alpha}\sqsubseteq \bdd{1}$.
The definition of $\bdd{1}$ coincides with the definition of the classic 
bisimilarity pseudometric $\mathsf{d}_1$ (see e.g.~\cite{chen2012complexity}), which satisfies $\mathsf{d}_1(s,s')=0$ if and only if $s$ and $s'$ are bisimilar.
Consequently, we obtain the following corollary.
\begin{corollary}
For any $\alpha\ge 1$,  if $s\sim s'$ then $\bd{\alpha}(s,s')=0$.
\end{corollary}
As in the case of $\tv{\alpha}$, we do not have the converse in our setting. Example~\ref{ex:kernel} shows that $s_0\not\sim s_1$ but we observe that $\bdd{1.5}(s_0,s_1)=0$. Observe:
\begin{multline*}\bdd{1.5}(s_0,s_1) \quad \le \\
\max_f \Big( \sum_{s \in S} f(s)(\mu_{s_0}(s)- 1.5\cdot \mu_{s_1}(s)),\sum_{s \in S} f(s)(\mu_{s_1}(s)- 1.5\cdot \mu_{s_0}(s))\Big) 
\\
= \max_f (f(s_2)(0.6 - 1.5\cdot 0.4) + f(s_3)(0.4-1.5\cdot 0.6), 
\\ \qquad \qquad\qquad  f(s_2)(0.4-1.5\cdot 0.6) + f(s_3)(0.6 - 1.5\cdot 0.4) ).
\end{multline*}
Notice the coefficients of $f(s)$ are all non-positive. Consequently, regardless of the restrictions on $f$, the maximising allocation will be $f(s) = 0$
and, thus,  $\bdd{1.5}(s_0,s_1) = 0$.

\subsection*{Example: Dining Cryptographers}
\begin{figure}[t] 
\centering
\begin{tabular}{c}
\begin{lstlisting}[label={lst:dc},captionpos=b,language=python]
diningCrypto(payingCryptographer):
	firstFlip = flip(p, 1-p)
	previousFlip = firstFlip
	for cryptographer = 0 $\to$ n-1:
		if cryptographer == n-1:
			thisFlip = firstFlip
		else:
			thisFlip = flip(p, 1-p)
		if (cryptographer == payingCryptographer):
			announce(previousFlip == thisFlip)
		else:
			announce(previousFlip != thisFlip)
		previousFlip = thisFlip
\end{lstlisting}
\end{tabular}
\caption{Simulation of Dining Cryptographers Protocol}
\label{fig:dc}
\end{figure}
In the dining cryptographer model~\cite{Chaum88}, a ring of diners want to determine whether one of the diners paid or an outside body.
If a  diner paid, we do not want to reveal which of them it was. The protocol proceeds with each adjacent pair privately flipping a coin, each diner then reports the XOR of the two coin flips they observe, however if the diner paid he would report the negation of this. We can determine if one of them paid by taking the XOR of the announcements. With perfectly fair coins, the protocol guarantees privacy of the paying diner, but it is still differentially private if the coins are biased. If an outside body paid, 
there is no privacy to maintain so we only simulate the scenarios in which one of the diners did pay. The scenario where Cryptographer 0 paid must have similar output distribution to Cryptographer 1 paying, so that it can be determined that one of them did pay, but not which. The internal configuration of the machine is always assumed to be hidden, but the announcements are made public whilst maintaining the privacy of the participating Cryptographer (and the internal states). 

The LMC in Figure~\ref{fig:twoplayer-dc} shows the 2-person dining cryptographers protocol (Figure~\ref{fig:dc}) starting from Cryptographers 0 and 1 using weighted coins with $p = \frac{49}{100}$. The states of the machine encode the 5 variables that need to be tracked. To achieve $\epsilon,\delta$-differential privacy with $\alpha = e^\epsilon = 1.0002$ the minimal (true) value of $\delta$ is $0.00030004$. Our methods generate a correct upper bound $\bd{\alpha}(s_0,s_1) = 0.0004$, showing $\ln(1.0002),0.0004$-differential privacy. The protocol could be played with $n$ players, requiring $O(n^2)$ states, for all possible assignments 
of paying cryptographer  and current cryptographer. In a two-person scenario, the diners would know which of them had paid but an external observer of the output would only learn that one of them paid, not which.

\begin{figure}[t]
\centering
\includegraphics[width=1\textwidth]{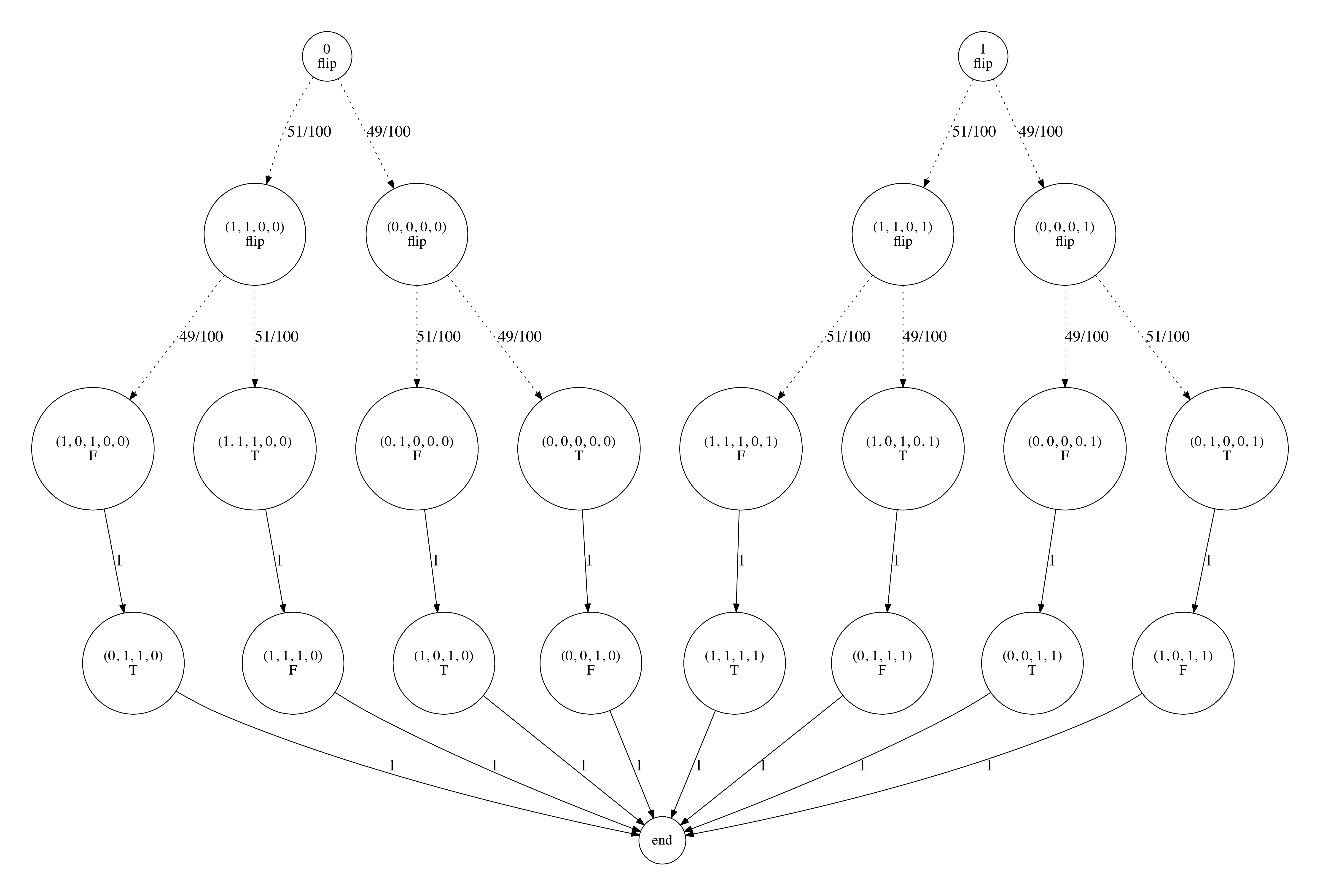}
\caption{Markov Chain for 2 dining cryptographers: state 0 (resp. 1) denotes Cryptographer 0 (resp. 1) paid. The first line of a node is the state name, the second line is the label of the state.}
\label{fig:twoplayer-dc}
\end{figure}

\section{Skewed Kantorovich distances\label{sec:kant}}

Here we discuss how to calculate our variant of the Kantorovich distance. This will inform the next section, in which we look into computing $\bd{\alpha}$.

Recall the definition of $\kanta(d)(\mu,\mu')$ from Definition~\ref{defn:sk}.
In the general case of $\Delta_\alpha (a,b)$, both $a-\alpha b$ and $b-\alpha a$ could be negative, so the maximum with $0$ is taken. 
However, within the Kantorovich function, the constant function $f(i) = 0$ is a valid assignment, which achieves $0$ in either case ($0 - \alpha \times 0 =0 $). 
Consequently, we can simplify the definition of $\Delta_\alpha$ to omit the $0$ case inside $K_\alpha$. 

If $\alpha=1$ then $\Delta_\alpha$ is the absolute value function and it is known that the distance corresponds to a single instance of a linear programming problem~\cite{van2001algorithm}.
However, this is no longer true in our case due to the shape of $\Delta_\alpha (x,y)=\max(x-\alpha y, y-\alpha x)$. 
Still, one can present the calculation as taking the maximum of a pair of linear programs.
We shall refer to this formulation as the ``primal form'' of $\kanta(d)$. We give the first program below,
the other is its symmetric variant with  $\mu, \mu'$ reversed. Below we write $f_i$ for $f(i)$ and let $i,j$ range over $S$, and assume that $d$ is symmetric.
\[
\max_{f \in [0,1]^{S}}\, \Big(\sum_{i} f_i \mu(i) -\alpha \sum_{i} f_i \mu'(i)\Big) \qquad \text{subject to}\qquad  \forall i,j\quad f_i - \alpha f_j \le d_{i,j}
\]

The standard Kantorovich distance ($\alpha=1$) is often presented in the following dual form when $m$ is a pseudometric, based on the minimum coupling between
the two distributions $\mu$ and $\mu'$, weighted by the distance function. 
\[
K(m)(\mu,\mu') = \min_{\omega \in [0,1]^{S\times S}} \sum_{i,j} 
\omega_{i,j} m_{i,j}\qquad\text{subject to} \qquad
\begin{array}{l}
\forall i\quad \sum_j \omega_{i,j} = \mu(i) \\
\forall j\quad \sum_i \omega_{i,j} = \mu'(j) \end{array}
\]

\begin{remark}
The dual form can be viewed as an optimal transportation problem in which an arbitrarily divisible cargo must be transferred from one set of locations (represented by a copy $S^L$ of $S$) to another (represented by a different copy $S^R$ of $S$).
Each state  $s^R\in S^R$ must receive $\mu(s)$, while each state $s^L\in S^L$ must send $\mu'(s)$. If $\omega_{i,j}$ is taken to represent the amount that gets sent from $j^L$ to $i^R$ then the above conditions restrict $\omega$ in accordance with the sending and receiving budgets.
If $d_{i,j}$ represents the cost of sending from $j^L$ to $i^R$ then the objective function  $\sum_{i,j} \omega_{i,j} \cdot d_{i,j}$ corresponds to the overall cost of transport. Consequently, the problem is referred to as 
a mass transportation problem~\cite{Kan42}.
\end{remark}
To achieve a similar ``dual form'' in our case,  we take the dual form of each of our linear programs.
Then we can calculate the distance by taking the maximum of the two minima.
The shape of the dual is given below on the right.
\begin{lemma}\label{thm:dualform}
\[
\begin{array}{rcl}
\displaystyle\max_{f \in [0,1]^{S}}\big(\sum_{i} f_i \mu(i) -\alpha \sum_{i} f_i \mu'(i)\big) 
&\quad =\quad & 
\displaystyle\min_{\omega \in [0,1]^{S\times S}, \tau,\gamma, \eta \in [0,1]^S} \sum_{i,j} \omega_{i,j}\cdot d_{i,j}  + \sum_{i} \eta_i
\\
\textup{subject to}& & \textup{subject to}
\\
\forall i,j \ f_i - \alpha f_j \le d_{i,j} & & \forall i: \sum_{j} \omega_{i,j}  + \tau_i  -\gamma_i + \eta_i= \mu(i)
\\  & & \forall j: \sum_i \omega_{i,j} + \frac{ \tau_j -\gamma_j}{\alpha} \le \mu'(j)
\end{array}
\]
\end{lemma}
The dual form presented above is a simplified (but equivalent) form of the immediate dual obtained via the standard LP recipe.
Note that the polytope we are optimising over is independent of $d$, which appears only in the objective function.
The dual of the other linear program is obtained by swapping $\mu,\mu'$.
\begin{remark}
In the skewed case,  we optimise over the following polytope
\[
\Omega_{\mu,\mu'} = \left\{ (\omega , \eta )\in [0,1]^{S\times S}\times[0,1]^S \quad | \quad   \begin{array}{l} \exists \gamma,\tau \in [0,1]^S \\ \quad \forall i: \sum_{j} \omega_{i,j} + \tau_i -\gamma_i +\eta_i = \mu(i)\\ \quad \forall j: \sum_i \omega_{i,j} + \frac{ \tau_j -\gamma_j}{\alpha} \le \mu'(j) \end{array} \right\}
\] 
One can also view it as a kind of transportation problem. As before, cargo can be transferred through the standard routes with $\omega$ at a cost $d$ or new resource $\eta$ can be obtained from an extra location at a cost of $1$. There are also additional, cost-free routes between corresponding pairs $s^L$ and $s^R$ (represented by $\tau_s$) and back (represented by $\gamma_s$). 
These extra routes are quite peculiar.  En route from $s^L$ to  $s^R$ the cargo `grows': when $\frac{\tau_s}{\alpha}$ is sent from $s^L$, a larger amount of $\tau_s$ is received at $s^R$.
Overall, the total amount of cargo sent may be less than that received, so the sending constraints are now inequalities. 
From $s^R$ to $s^L$ the cargo `shrinks': when $\gamma_s$ is sent from $s^R$, only $\frac{\gamma_s}{\alpha}$ is received by $s^L$. 

It is immediate that $\tau$ routes can be useful.
The $\gamma$ routes may be useful for optimisation under two conditions. Firstly 
the shrinkage of the cargo must be made up elsewhere, i.e.,
through `growing' $\tau$ routes. Additionally the cost $\alpha \times d(s_1^L, s^R) + d(s^L, s_2^R)$ is lower than $d(s_1^L, s_2^R)$, which may well be the case due to the lack of triangle inequality.
\end{remark}

We arrive at the following formulation, which we call the ``dual form'':
$$K_\alpha(d)(\mu, \mu') = \max \left\{ \min_{\omega, \eta \in \Omega_{\mu,\mu'}}\sum_{i,j} \omega_{i,j}\cdot d_{i,j} + \sum_i \eta_i
 , \ \min_{\omega,\eta \in \Omega_{\mu',\mu}}\sum_{i,j} \omega_{i,j}\cdot d_{i,j}  + \sum_i \eta_i  \right\}.
$$
Note that $K_\alpha(d)(\mu,\mu')$ can  be computed in polynomial time as a pair of linear programs in either primal or dual form, 
and taking the maximum (in either case).
In our calculations related to $\bd{\alpha}$, the distributions $\mu,\mu'$ will always be taken to be $\mu_s,\mu_{s'}$ respectively, for some $s,s'\in S$.
The ability to switch between primal and dual form will play a useful role in our complexity-theoretic arguments.

\section{Computing $\bd{\alpha}$}\label{sec:comp}

We start off by observing that all distances $\bd{\alpha}(s,s')$ are rational and can be expressed in polynomial size with respect to $\mcal{M}$.
To that end, we exploit a result by Sontag \cite{sontag1985real}, which states that, without affecting satisfiability, 
quantification in the first-order fragment of linear real arithmetic (LRA) 
can be restricted to rationals of polynomial size  with respect to formula length (as long as all coefficients present in the formula are rational). 
Consequently, if we can express ``there exists a least fixed point $d$ of $\ga$" in this fragment (with a polynomial increase in size), we can draw the intended conclusion.

 We give the relevant formula in Figure~\ref{fig:form:sigmatwop}.
The formula asserts the existence of a distance $d$,
which is a pre-fixed point of $\Gamma_\alpha$  ($\forall f. \phi(d,f)$) such that any other pre-fixed point $d'$ of $\Gamma_\alpha$ is greater.
Note that $\forall f.\phi(f,d)$ exploits the fact that $\max_f A(f) \le d(s,s')$ is equivalent to $\forall f (A(f)\le d(s,s'))$. Sontag's result
then implies the following.

\begin{theorem}\label{thm:polyrational}
Values of $\bd{\alpha}$ are rational.  There exists a polynomial $p$ such that for any LMC $\mcal{M}$ and 
$s,s'\in S$, the size of $\bd{\alpha}$ (in binary) can be bounded from above by a polynomial in $|\mcal{M}|$. 
\end{theorem}

\begin{remark}
\label{rem:poly}
Sontag~\cite{sontag1985real} uses the fact mentioned above to relate the alternation hierarchy within LRA to the polynomial hierarchy $\PH$:
formulae of the form $\exists x_1 \forall x_2 \dots Q x_k F(x_1 \dots x_k)$ (with quantifier-free $F$) 
correspond to $\Sigma_k^\P$ (and formulae starting with $\forall$ to $\Pi_k^\P$). Recall that $\Sigma_1^\P=\NP$.
\end{remark}
\begin{figure}[t]
\begin{align*}
 \exists d \in [0,1]^{S\times S} &\ (\forall f \in [0,1]^{S} \phi(d,f) \wedge \forall d' \in [0,1]^{S\times S} (\forall f \in [0,1]^{S} \phi(d',f) \implies \bigwedge_i d_i \le d'_i))
\\ \phi(d,f) &= \bigwedge_{s,s'} 
\begin{cases}
    d_{s,s'} = 1 &  \lab(s) \ne \lab(s') \\
    (\bigwedge_{i,j} f_i -\alpha f_j \le d_{i,j} \wedge f_j -\alpha f_i \le d_{i,j} ) & \lab(s) = \lab(s') \\ 
    \quad \implies (\sum_{i} f_i \mu_s(i) -\alpha \sum_{i} f_i \mu_{s'}(i) \le d_{s,s'} \\ \quad \quad\quad\quad \wedge \sum_{i} f_i \mu_{s'}(i) -\alpha \sum_{i} f_i \mu_{s}(i) \le d_{s,s'} )
\end{cases} 
\end{align*}
\caption{Logical formulation of least pre-fixed point.}
\label{fig:form:sigmatwop}
\end{figure}
Next we focus on the following decision problem for $\bd{\alpha}$.

\begin{center}
\bdp: given $s,s'\in S$ and $\theta\in \Q$, is it the case that $\bd{\alpha}(s,s') \le \theta$?
\end{center}
Recall that  the analogous problem for $\tva$ is undecidable (Remark~\ref{rem:bad}).
In our case, the problem turns out to be decidable and the argument does not depend on whether  $<$ or $\le$ is used.
To establish decidability we can observe that $\bd{\alpha}(s,s') \le \theta$ can be expressed in LRA simply by adding $d(s,s')\le \theta$ to the formula from Figure~\ref{fig:form:sigmatwop}.

We can simplify the formula, though,  using  $\bd{\alpha}=\min\,{\{d\,|\, \Gamma_\alpha(d)\sqsubseteq d\}}$.
Then $\bd{\alpha}(s,s')\le \theta$  can be specified as  the existence of a pre-fixed point $d$ such that $d(s,s') \le\theta$.
This can be done as follows, using $\phi(d,f)$ from Figure~\ref{fig:form:sigmatwop}.
\[
\exists d \in [0,1]^{S\times S}\, (\ \forall f \in [0,1]^{S} \phi(d,f)\ \wedge\ d(s,s')\le \theta\ )
\]
By Sontag's results, this not only yields decidability but also membership in $\Sigma_2^\P$. Recall that $\NP\subseteq \Sigma_2^\P\subseteq \PH\subseteq \PSPACE$.

Note that the universal quantification over $f$ remains, i.e. we can still only conclude that the problem is in $\Sigma_2^\P$.
To overcome this, we shall use the dual form instead (Lemma~\ref{thm:dualform}). This will enable us to eliminate the universal quantification and replace it
with existential quantifiers using the fact that
$\min_\omega A(\omega)\le B$ is equivalent to $\exists \omega (A(\omega)\le B)$. The resultant formula is shown in Figure~\ref{fig:form:less}.

\begin{figure}[t]
\begin{align*}
    \mathit{\bdp}(s,s',\theta) =& \exists (d_{i,j})_{i,j  \in S} \quad \bigwedge_{i,j  \in S} (0\le d_{i.j} \le 1) \ \wedge \  \mathit{prefixed}(d)\  \wedge \  d_{s,s'} \le \theta
    \\
    \mathit{prefixed}(d) =& \bigwedge_{q, q' \in S} \begin{cases}
     d_{q,q'} = 1 & \lab(q) \ne \lab(q') \\
     \mathit{prefixed}_1(d,d_{q,q'},q,q')  & \lab(q)  =\lab(q')
    \\ \quad\quad \wedge \     \mathit{prefixed}_1(d,d_{q,q'},q',q) 
    \end{cases} \\
    \mathit{prefixed}_1(d,x, q,q') = &\exists (\omega_{i,j})_{i,j  \in S} 
    \quad \exists (\gamma_{i})_{i \in S}
    \quad \exists (\tau_{i})_{i \in S}
    \quad \exists (\eta_{i})_{i \in S}
    \\ & \quad \sum_{i,j  \in S} \omega_{i,j}\cdot d_{i,j} + \sum_i \eta_i \le x
    \\ & \wedge \bigwedge_{i,j  \in S} (0 \le \omega_{i,j} \le 1)\wedge \bigwedge_{i \in S} (0 \le\gamma_{i} \le 1\, \wedge\, 0 \le \tau_{i} \le 1 \wedge\, 0 \le \eta_{i} \le 1)
    \\ & \wedge \bigwedge_{i \in S}( \sum_{j \in S} \omega_{i,j} -\gamma_i+ \tau_i +\eta_i = \mu_{q}(i))  \wedge \bigwedge_{j \in S} ( \sum_{i \in S} \omega_{i,j} + \frac{ \tau_j -\gamma_j}{\alpha} \le \mu_{q'}(j))
\end{align*}
\caption{$\NP$ Formula for \bdp}
\label{fig:form:less}
\end{figure}

Note the formula is not linear due to $\omega_{i,j}\cdot d_{i,j}$. However, because we know (Theorem~\ref{thm:polyrational}) 
that $\bd{\alpha}$ corresponds to an assignment of poly-sized rationals,
 we can consider the formula with $d$ fixed at $\bd{\alpha}$. Then it does become an LRA formula (of polynomially bounded length with respect to $|\mcal{M}|$) and
we can again conclude that  the assignments of $\omega, \gamma, \tau$  must also involve rationals whose size is polynomially bounded. 
 Consequently, the formula implies membership of our problem in $\Sigma_1^\P=\NP$: it suffices to guess  the satisfying assignment, guaranteed to be rational and of polynomial size.
\begin{theorem}\label{thm:bdp}
$\bdp$ is in $\NP$.
\end{theorem}

The decidability of $\bdp$ makes it possible to approximate $\bd{\alpha}(s,s')$ to arbitrary (rational) precision $\epsilon$ by binary search.
This will involve $O(|\epsilon|)$ calls to the oracle for $\bdp$ (where $|\epsilon|$ is the number of bits required to represent $\epsilon$ in binary).

What's more, assuming the oracle, one can actually find the exact value of $\bd{\alpha}(s,s')$ in polynomial time (wrt $\mcal{M}$). This exploits the fact that the value of $\bd{\alpha}$ is rational and its size is polynomially bounded, so one can find it by approximation
to a carefully chosen level of precision and then finding the relevant rational with the continued fraction algorithm~\cite{GLS88,EY10}.

\begin{theorem}\label{thm:main}
$\bd{\alpha}$ can be calculated in polynomial time with an $\NP$ oracle.
\end{theorem}

As a consequence, the problem of computing $\bd{}$
reduces to propositional satisfiability, i.e., can be encoded in SAT.
This justifies, for instance, the following approach: treat every variable
as a ratio of two integers from an exponential range, and give the system
of resulting constraints to an Integer Arithmetic or SAT solver. While
this might look like resorting to a general-purpose ``hammer'', Theorem~\ref{thm:main}
is necessary for this method to work: it is not, in fact, possible to solve general polynomial constraint systems
relying just on SAT.%
\footnote{%
More precisely, the existence of such a procedure would be a breakthrough
in the computational complexity theory,
showing that $\mathbf{NP} = \exists\mathbb R$. This would imply that
a multitude of problems in computational geometry could be solved using
SAT solvers~\cite{SchaeferS17,Cardinal15}. Unlike for $\bd{}$, variable assignments in these problems may need to be
irrational, even if all numbers in the input data are integer or rational.
}

We expect, however, this direct approach to be inferior to the following
observation. Theorem~\ref{thm:polyrational} reveals that the variables in our constraint
system need not assume irrational values or have large bit
representations. Thus, one can give the system to a more
powerful theory solver, or an optimisation tool, but to expect that the
existence of simple and small models (solutions) will help the SMT
heuristics (resp.\ optimization engines) to find them quickly.

\section{Conclusion and Further Work}

We have demonstrated that bisimilarity distances can be used to determine differential privacy parameters, despite their non-metric properties. 
We have established that the complexity of finding these values is polynomial, relative to an $\NP$ oracle. Yet, 
it may still be possible to obtain a polynomial algorithm---although much like in the case of the classical bisimilarity distances and linear programming, it may not necessarily outperform theoretically slower procedures.

We conjecture that $\bd{\alpha}$, which we defined as the least fixed point of the operator $\Gamma_\alpha$, may in fact be characterized as the unique fixed point of a similar operator. By the results of Etessami and Yannakakis~\cite{EY10}, it would then follow that $\bd{\alpha}$ can be computed in $\PPAD$, a smaller complexity class, improving upon our $\NP$ upper bound and matching the complexity of a closely related setting (see below). The reason is the continuity of $\Gamma_\alpha$, which follows from the properties of the polytope over which $f$ ranges (in the definition of $K_\alpha(d)$). Whether $\bd{\alpha}$ can in fact be computed in polynomial time or is $\PPAD$-hard seems to be a challenging open question.

Our existing work is limited to labelled Markov chains, or fully probabilistic automata. However, the standard bisimulation distances can also be defined on deterministic systems, where 
their computational complexity is $\PPAD$ \cite{van2014complexity}. In our scenario, the privacy can only be analysed between two start states, but it is also reasonable to allow an input in the form of a trace or sequence of actions; the output would also be a trace. Here the choice of labels (at a specific state) would correspond to decisions taken by the user, and the availability of only one label would mean that this is the output. This setting would support a broader range of scenarios that could be modelled and verified as differentially private.
 
\subsubsection*{Acknowledgement}
David Purser gratefully acknowledges funding by the UK Engineering and Physical Sciences Research Council (EP/L016400/1), the EPSRC Centre for Doctoral Training in Urban Science. Andrzej Murawski is supported by a Royal Society Leverhulme Trust Senior Research Fellowship and the International Exchanges Scheme (IE161701).

\bibliographystyle{splncs04}
\bibliography{atva}

\begin{thebibliography}{10}
\providecommand{\url}[1]{\texttt{#1}}
\providecommand{\urlprefix}{URL }
\providecommand{\doi}[1]{https://doi.org/#1}

\bibitem{albarghouthi2017synthesizing}
Albarghouthi, A., Hsu, J.: Synthesizing coupling proofs of differential
  privacy. Proceedings of the ACM on Programming Languages  \textbf{2},
  58:1--58:30 (2018)

\bibitem{bacci2013fly}
Bacci, G., Bacci, G., Larsen, K.G., Mardare, R.: On-the-fly exact computation
  of bisimilarity distances. In: TACAS. pp. 1--15. Springer (2013)

\bibitem{baier2008principles}
Baier, C., Katoen, J.P.: Principles of model checking. MIT Press (2008)

\bibitem{barthe2015relational}
Barthe, G., Espitau, T., Gr{\'e}goire, B., Hsu, J., Stefanesco, L., Strub,
  P.Y.: Relational reasoning via probabilistic coupling. In: LPAR. pp.
  387--401. Springer (2015)

\bibitem{barthe2012probabilistic}
Barthe, G., K{\"o}pf, B., Olmedo, F., Zanella~B{\'e}guelin, S.: Probabilistic
  relational reasoning for differential privacy. In: POPL. pp. 97--110. ACM
  (2012)

\bibitem{bill86}
Billingsley, P.: Probability and Measure. John Wiley and Sons, 2nd edn. (1986)

\bibitem{van2017probabilistic}
van Breugel, F.: Probabilistic bisimilarity distances. ACM SIGLOG News
  \textbf{4}(4),  33--51 (2017)

\bibitem{van2007approximating}
van Breugel, F., Sharma, B., Worrell, J.: Approximating a behavioural
  pseudometric without discount. In: FoSSaCS. pp. 123--137. Springer (2007)

\bibitem{van2001algorithm}
van Breugel, F., Worrell, J.: An algorithm for quantitative verification of
  probabilistic transition systems. In: CONCUR. pp. 336--350. Springer (2001)

\bibitem{van2014complexity}
van Breugel, F., Worrell, J.: The complexity of computing a bisimilarity
  pseudometric on probabilistic automata. In: Horizons of the Mind. A Tribute
  to Prakash Panangaden, LNCS, vol.~8464, pp. 191--213. Springer (2014)

\bibitem{Cardinal15}
Cardinal, J.: Comput. geometry column 62. {SIGACT} News  \textbf{46}(4),
  69--78 (2015)

\bibitem{chatzikokolakis2014generalized}
Chatzikokolakis, K., Gebler, D., Palamidessi, C., Xu, L.: Generalized
  bisimulation metrics. In: CONCUR. pp. 32--46. Springer (2014)

\bibitem{Chaum88}
Chaum, D.: The dining cryptographers problem: Unconditional sender and
  recipient untraceability. J. Cryptology  \textbf{1}(1),  65--75 (1988)

\bibitem{chen2012complexity}
Chen, D., van Breugel, F., Worrell, J.: On the complexity of computing
  probabilistic bisimilarity. In: FoSSaCS. pp. 437--451. Springer (2012)

\bibitem{deng2009kantorovich}
Deng, Y., Du, W.: The {K}antorovich metric in computer science: A brief survey.
  Electronic Notes in Theoretical Computer Science  \textbf{253}(3),  73--82
  (2009)

\bibitem{desharnais2004metrics}
Desharnais, J., Gupta, V., Jagadeesan, R., Panangaden, P.: Metrics for labelled
  markov processes. Theoretical computer science  \textbf{318}(3),  323--354
  (2004)

\bibitem{desharnais2002metric}
Desharnais, J., Jagadeesan, R., Gupta, V., Panangaden, P.: The metric analogue
  of weak bisimulation for probabilistic processes. In: LICS. pp. 413--422.
  IEEE (2002)

\bibitem{dwork2006calibrating}
Dwork, C., McSherry, F., Nissim, K., Smith, A.: Calibrating noise to
  sensitivity in private data analysis. In: TCC. pp. 265--284. Springer (2006)

\bibitem{EY10}
Etessami, K., Yannakakis, M.: On the complexity of {N}ash equilibria and other
  fixed points. {SIAM} J. Comput.  \textbf{39}(6),  2531--2597 (2010)

\bibitem{GLS88}
Gr{\"o}tschel, M., Lov{\'a}sz, L., Schrijver, A.: {Geometric Algorithms and
  Combinatorial Optimization}, Algorithms and Combinatorics, vol.~2. Springer
  (1988)

\bibitem{Kan42}
Kantorovich, L.V.: On the translocation of masses. Doklady Akademii Nauk SSSR
  \textbf{37(7-8)},  227–--229 (1942)

\bibitem{Kie18}
Kiefer, S.: On computing the total variation distance of hidden markov models.
  In: ICALP, pp. 130:1--130:13 (2018)

\bibitem{larsen1991bisimulation}
Larsen, K.G., Skou, A.: Bisimulation through probabilistic testing. Information
  and computation  \textbf{94}(1),  1--28 (1991)

\bibitem{SchaeferS17}
Schaefer, M., Stefankovic, D.: Fixed points, {N}ash equilibria, and the
  existential theory of the reals. Theory Comput. Syst.  \textbf{60}(2),
  172--193 (2017)

\bibitem{sontag1985real}
Sontag, E.D.: Real addition and the polynomial hierarchy. IPL  \textbf{20}(3),
  115--120 (1985)

\bibitem{tang2016computing}
Tang, Q., van Breugel, F.: Computing probabilistic bisimilarity distances via
  policy iteration. In: CONCUR. pp. 22:1--22:15. Leibniz-Zentrum (2016)

\bibitem{tarski1955lattice}
Tarski, A.: A lattice-theoretical fixpoint theorem and its applications.
  Pacific Journal of Mathematics  \textbf{5}(2),  285--309 (1955)

\bibitem{tschantz2011formal}
Tschantz, M.C., Kaynar, D., Datta, A.: Formal verification of differential
  privacy for interactive systems. ENTCS  \textbf{276},  61--79 (2011)

\bibitem{V17}
Vadhan, S.P.: The complexity of differential privacy. In: Tutorials on the
  Foundations of Cryptography, pp. 347--450. Springer (2017)

\bibitem{xu2015formal}
Xu, L.: Formal Verification of Differential Privacy in Concurrent Systems.
  Ph.D. thesis, Ecole Polytechnique (Palaiseau, France) (2015)

\bibitem{xu2014metrics}
Xu, L., Chatzikokolakis, K., Lin, H.: Metrics for differential privacy in
  concurrent systems. In: FORTE. pp. 199--215. Springer (2014)

\end{thebibliography}
\newpage
\appendix

\appendix
\section{Proving Lemma~\ref{lem:bounds}. $\tv{\alpha} \sqsubseteq \bd{}$}

We will first introduce the skewed Kantorovich distance $K_\alpha(\dm)$ over traces that we will place in between the two. We will show that for all $s,s' \in S$: $\tv{\alpha}(s,s') = K_\alpha(\dm)(\nu_s,\nu_{s'}) \le \bd{}(s,s;)$. We consider each of these (in)equalities:
\begin{lemma}\label{lem:appen:eqpart}
$\tv{\alpha}(s,s') = K_\alpha(\dm)(\nu_s,\nu_{s'})$
\end{lemma}

\begin{lemma}\label{lem:appen:mainbound}
$K_\alpha(\dm)(\nu_s,\nu_{s'}) \le \bd{}$
\end{lemma}

\subsection{Definition and Properties of $K_\alpha(\dm)(\nu_s,\nu_{s'})$}
We introduce a generic version of the skewed Kantorovich distance, defined for any object $X$ rather than states $(X = S)$ as defined in the main body.

Given $\mu,\mu' \in Dist(X)$ and $d:X\to[0,1]$ a distance function 
$$K_\alpha(d)(\mu, \mu') = \sup_{\substack{f: X \to [0,1] \\  \forall x,x' \in X \Delta_\alpha(f(x),f(x')) \le d(x,x')  }} \Delta_\alpha(\int f d\mu,\int f d\mu') $$

We can instantiate this on distributions for traces with $X = \alphabet{}^\omega$ so can measure difference between trace distributions for states e.g. $\nu_s,\nu_{s'}$. 

We assume whenever we write $f$ that we restrict only to those which are measurable in our space $(\alphabet^\omega, \mcal{F})$.

Let us define
$\dm(t,t') : \alphabet{}^\omega \times \alphabet{}^\omega \to \{0,1\}$ as the indicator function which is one if the arguments are not the same and zero if they are. We also define $\dm^h$ a restriction considering only the prefix of length $h$, ($t^h$ is the prefix of length $h$ of trace $t$).

\[
\dm(t, t' )= \begin{cases} 1 & \text{if } t \ne t' \\ 0 & \text{otherwise} \end{cases} 
\quad 
\text{ and }
\quad 
\dm^h(t, t' )= \begin{cases} 1 & \text{if } t^h \ne t'^h \\ 0 & \text{otherwise} \end{cases} 
\]

\begin{remark} 
The use of $\dm$ is so that $\Delta_\alpha(f(t),f(t')) \le \dm(t,t)$ is rather no restriction at all, since either a trace is the same, thus $\Delta_\alpha(f(t),f(t')) = 0$ or the traces are different and $\dm(t,t') = 1$. So $K_\alpha(\dm)$ takes supremum over all measurable $f$.
\end{remark}

\begin{remark}
In the following sections we will use the notation $\hat{f}(\mu) = \int f d\mu$.
\end{remark}

\subsection{Proof of Lemma~\ref{lem:appen:eqpart}: $\tv{\alpha}(s,s') = K_\alpha(\dm)(\nu_s,\nu_{s'})$}
Whilst we only really require $\tv{\alpha}(s,s') \le K_\alpha(\dm)(\nu_s,\nu_{s'})$ for our main proof we will show equality, this will give us an alternative representation of $\tv{\alpha}$ that will prove useful along the way.

\begin{claim}
$\tv{\alpha}(s,s') \le K_\alpha(\dm)(\nu_s,\nu_{s'})$
\end{claim}
\begin{proof}
Consider $E \in \mcal{F}$ then $\ind{E}$ is a measurable function in the argument of the supremum of $K_\alpha(\dm)$ with $\int \ind{E} d\nu_s = \nu_s(E)$ thus $\Delta_\alpha(\nu_s(E), \nu_{s'}(E)) = \Delta_\alpha(\int \ind{E} d\nu_s,\int \ind{E} d\nu_{s'})$.
\end{proof}

\begin{claim}
 $\tv{\alpha}(s,s') \ge K_\alpha(\dm)(\nu_s,\nu_{s'})$
 \end{claim}
\begin{proof}
\label{proof:tvisk}
We have shown a 1-1 correspondence between indicator functions and events. We further show any function that is not an indicator function does not contribute to the supremum. We proceed first by considering simple functions and extend this to other functions with the Monotone convergence theorem. Our proof follows the technique of \cite{xu2015formal}, however we define our new function explicitly.

$$ K_\alpha(\dm)(\mu, \mu')= \max \sup_{f: A^\omega \to [0,1]}\{\hat{f}(\mu) - \alpha \hat{f}(\mu') \},  \sup_{f: A^\omega \to [0,1]}\{\hat{f}(\mu') - \alpha \hat{f}(\mu) \}$$

This is symmetric, so without loss of generality we consider only $\hat{f}(\mu) - \alpha \hat{f}(\mu')$ as our distance.

A simple function $f$ can be described by a finite sum $\sum_i a_i \ind{A_i}$, where $\ind{A_i}$ is the characteristic function on $A_i \in \mcal{F}$. This means $f$ takes on finitely many values in its range.

We want to show that if $f(a) \in (0,1)$ then we can only increase the value of $\hat{f}(\mu) - \alpha \hat{f}(\mu')$ by changing to $f: A^\omega \to \{0,1\}$. Note this will maintain that $f$ valid in the restriction of the sup, since we only require to stay measurable.

Since we have $f$ simple then $img(f)$ is finite. Let $n = |img(f) \setminus \{0,1\}|$

Consider $v \in img(f_i) \setminus \{0,1\}$. Let $A_i = \{x : f(x) = v\}$

Then we can define $f_{i+1} = f_{i} + g_{i}$ with $g_{i} = t \times \mathbb{1}_{A_i}$ where $t$ is defined as follows:

if $\mu(A_i) - \alpha \mu'(A_i) \ge 0$ we set $t = 1-v$, this gives us $f_{i+1}(x) = 1$ for all $x \in A_i $. Otherwise $t = -v$ so $ f_{i+1}(x) = 0$ for all $x \in A_i$.

We show $f_{i+1}$ gives no smaller value in the supremum.
\begin{align*}
\hat{f}_{i+1}(\mu) - \alpha \hat{f}_{i+1}(\mu') &= \int f_{i+1}d(\mu) - \alpha \int  f_{i+1} d(\mu')\\
&= \int f_{i} + g_i d(\mu) - \alpha \int  f_{i+1} + g_i d(\mu') \\
&= \int f_{i}  d(\mu) - \alpha \int f_{i}d(\mu') +  \int g_i d(\mu) - \alpha  \int  g_i d(\mu') \\ 
&= \hat{f}_{i}(\mu) - \alpha \hat{f}_{i} (\mu') +  t \mu(A_i) - \alpha t \mu'(A_i)  \\
&\ge \hat{f}_{i}(\mu) - \alpha \hat{f}_{i} (\mu') \tag{Either $t>0 $ and $\mu(A_i) - \alpha \mu'(A_i) \ge 0$ or $t<0 $ and $\mu(A_i) - \alpha \mu'(A_i) < 0$}
\end{align*}
In particular the indicator $f_n$  is not worse than the simple function $f$, that is  
$$\hat{f}_{n}(\mu) - \alpha \hat{f}_{n}(\mu') \ge \hat{f}(\mu) - \alpha \hat{f}(\mu').$$

Consider for $f$ not simple, which can be approximated by $h_1, h_2, \dots$ simple, converging point-wise to $f$. Then by monotone convergence principle, with $d(a,b) = a - \alpha b$ continuous.

$$\lim_{n \to \infty} \hat{h}_n(\mu) - \alpha \hat{h}_n(\mu') = \hat{f}(\mu) - \alpha \hat{f}(\mu').$$

For every simple $h_i$, we are not smaller than an indicator function. Since it holds for each $n$ we are no larger than an indicator function, the limit is also no larger than the supremum of indicator functions. Therefore indicator functions, which have 1-1 correspondence with events are sufficient in the supremum.
\qed
\end{proof}

\begin{corollary}\label{cor:indicator}
When using $\dm$ over trace distributions we need only consider the supremum over indicator functions with no restrictions.
$$K_\alpha(\dm)(\mu, \mu') = \sup_{\substack{f: X \to \{0,1\} }} \Delta_\alpha(\int f d\mu,\int f d\mu') $$
\end{corollary}

\subsection{Proving Lemma~\ref{lem:appen:mainbound}: $K_\alpha(\dm)(\nu_s,\nu_{s'}) \le \bd{}$}

Our proof strategy generally follows the strategy of \cite{xu2015formal}. We will show for all prefixes of a trace $h$ restricting $\dm$ to the prefix will provide the required bound for all $h$. The main change in this argument is different in the base case, where an additional result (Lemma~\ref{lem:expectzero}) 
is needed to compensate for the fact that $\Delta_\alpha$ is not a metric. The induction step is similar, changing only our distance function $\Delta_\alpha$.

We will then extend $\dm^h$ to the supremum over $h$, i.e. $K_\alpha(\dm)$. To do this, we will approximate the events in $K_\alpha(\dm)$ by cylinders and show that for some $h$, $K_\alpha(\dm^h)$ is $\epsilon$ close to $K_\alpha(\dm)$. An additional lemma to support this is provided in Lemma~\ref{lem:genset} and the extension is shown in Lemma~\ref{lem:indenough}.

This extension from $K_\alpha(\dm^h)$ to $K(\dm)$ differs from the strategy of \cite{xu2015formal} who argue it is enough to show continuity of $K_V(m)$ with respect to $m$. Continuity \textit{does} hold in our case, but it is unclear to us that this actually shows the result. Since $\dm^h$ are discrete the $\epsilon, \delta$ formulation of continuity says at $\dm$ there exists $m$ such that $\max_{a,b} |\dm(a,b)  - m(a,b)| \le \delta$, however no such $m = \dm^h$ satisfies this unless it is exactly $\dm$ as if it differs in any point, the difference is 1.

\subsubsection{Additional Lemmas}

Note that in our case, given $\Delta_\alpha(f(x),f(x'))=0$ we cannot conclude $f(x)=f(x')$, because $\Delta_\alpha$ is not a metric. 
To compensate for this, we show a weaker result that still holds.
\begin{lemma}
\label{lem:expectzero}
Consider $f:X\to [0,1]$ such that $\forall x,x': \Delta_\alpha(f(x),f(x')) = 0 $.
Then $\forall \mu, \mu': \Delta_\alpha(\hat{f}(\mu),\hat{f}(\mu')) = 0$.
\end{lemma}

\begin{proof}

We consider the range of values $f$ could take, first consider $a = \inf_x \{f(x)\}$

We know that for all $x$, $\Delta_\alpha(f(x),a) \le 0$. So $f(x)-\alpha a \le 0$ and $a-\alpha f(x) \le 0$.

$f(x) \le \alpha a$ and we known $f(x)\ge a$ by definition so $f(x) \in [a, \alpha a] $ for all $x$.

Now consider $\int_{X} f d\mu $ where $f(x) \in [a, \alpha a]$ and $\int_{X} d\mu = 1$. 
\begin{align*}
\hat{f}(\mu) &= \int_{X} f d\mu 
\ge \int_{X} a d\mu = a \int_{X} d\mu
= a\\
\hat{f}(\mu) &= \int_{X} f d\mu
\le \int_{X} \alpha a d\mu 
=  \alpha a \int_{X}  d\mu
= \alpha  a
\end{align*}

So for all $x$ and $\mu$ then $\hat{f}(\mu) = \int_{X} f d\mu \in [a, \alpha a]$, telling us the expectation must also lie in this range. Next we notice that any two numbers in this range give distance zero, in particualr the expections.

Consider $x,y \in [a, \alpha a]$ 
\begin{align*}
    \Delta_\alpha(x,y) &\le \max(x - \alpha y, y - \alpha x, 0) \\
    &\le \max(\alpha a -  \alpha a, \alpha a - \alpha a, 0) \\
    &= \max(0,0,0) = 0
\end{align*}

Therefore $\Delta_\alpha(\hat{f}(\mu),\hat{f}(\mu_s)) = 0$.

\qed
\end{proof}

The following lemma generalises the classic result that measure on any measurable event can be approximated by events from the generating set. In our case, the generating set (algebra) corresponds to finite unions of cylinders, which are themselves determined by sequences from $\Sigma^\ast$.
We show that simultaneous approximation to the same degree of accuracy is  possible for two different measures too.

\begin{lemma} 
\label{lem:genset}
Let $(B, \mcal{B}, \mu)$ and $(B, \mcal{B}, \mu')$ be measure spaces over the $\sigma$-algebra $(B, \mcal{B})$.  Let $\mcal{A} \subset \mcal{B} $ be an algebra generating $\mcal{B}$.
$S = \{ X \in \mathcal B \quad | \quad  \forall \epsilon \  \exists A \in \mathcal A \text{ such that } \mu(A \triangle X)<\epsilon \text{ and } \mu'(A \triangle X)< \epsilon \}$ also forms a $\sigma$-algebra.
\end{lemma}
\begin{proof}
We adapt the proof from \url{https://math.stackexchange.com/questions/228998/approximating-a-sigma-algebra-by-a-generating-algebra} to the case where there are two measures.

To show that a set $S$ is a $\sigma$-algebra we require the inclusion of the empty set, closure under complement, and closure under countable unions. We show the full set rather than the empty set in Case 1, which shows the empty set by complement in Case 2. We show finite unions in Case 3 and extend to countable union in Case 4. 

\begin{case}
Since $B \in \mcal{A}$, $B\in \mcal{S}$.
\end{case}
\begin{case} Complement:
If $X \in \mcal{S}$ and $\epsilon > 0$ then there $\exists A \in \mathcal A \text{ such that } \mu(A \triangle X)<\epsilon \text{ and } \mu'(A \triangle X)< \epsilon$

Then $A^c \in \mcal{A}$ and $\mu(A^c \triangle X^c) = \mu(A \triangle X) <\epsilon$, so $X^c \in S$
\end{case}

\begin{case}Finite Union:
Let $X_1, X_2 \in S$ then  $\exists A_i \in \mathcal A \text{ such that } \mu(A_i \triangle X_i)<\frac{\epsilon}{2} \text{ and } \mu'(A_i \triangle X_i)< \frac{\epsilon}{2}$.

$\mu(X_1 \cup X_2 \triangle A_1 \cup A_2) \le \mu(X_1 \cup A_1 \triangle X_2 \cup A_2) \le \frac{\epsilon}{2} + \frac{\epsilon}{2}  = \epsilon$, similarly for $\mu'$ and $A_1 \cup A_2 \in \mcal{A}$ so $X_1 \cup X_2 \in S$.
\end{case}

\begin{case} Countable Union.
Let $\{X_k\} \subset S$, pairwise disjoint and $\epsilon > 0$. For each $k$, let $A_k \in \mcal{A}$ such that $\mu(A_k \triangle X_k) \le \frac{\epsilon}{2^k}$ and $\mu'(A_k \triangle X_k) \le \frac{\epsilon}{2^k}$

Take N such that $\mu(\bigcup_{j > N}X_j) \le \frac{\epsilon}{2}$ and $\mu'(\bigcup_{j > N}X_j) \le \frac{\epsilon}{2}$. (Make $\mu(\bigcup_{j > N}X_j) \le \sum_{j > N} \mu(X_j)$ arbitrarily small due to finite  measure, take $N$ big enough so both are $\mu$ and $\mu'$ sufficiently small).

Let $A = \bigcup\limits_{j=1}^{N} A_j \in \mcal{A}$.

Then $(\bigcup_k X_k)\triangle A \subset \bigcup\limits_{j=1}^{N}(X_j\triangle A_j) \cup \bigcup_{j > N} X_j$

Then $\mu(\bigcup_k X_k)\triangle A ) \le \mu(\bigcup\limits_{j=1}^{N}(X_j\triangle A_j) \cup \bigcup_{j > N} X_j) \le \sum_{j=1}^{N} \frac{\epsilon}{2^j} + \frac{\epsilon}{2} \le \epsilon$ and also for $\mu'$.

Then $\bigcup_k X_k \in S$

\end{case}
\qed
\end{proof}

Our stratergy will be to prove the result by induction, we show this is sufficient to extend to $K_\alpha(\dm)$ in the following lemma, using Lemma~\ref{lem:genset}.

\begin{lemma}
\label{lem:indenough}
We show if $K_\alpha(\dm^h) \le \bd{}$ for all $h$ then $K_\alpha(\dm) \le \bd{\alpha}$.
\end{lemma}

\begin{proof}

We now extend to the result from all $h$ to $K_\alpha(\dm)$ by showing that $\forall \epsilon \ \exists h$ such that $K_\alpha(\dm) - K_\alpha(\dm^h) \le \epsilon$ (notice $K_\alpha(\dm)$ is always larger).

Either 
$$K_\alpha(\dm)(\nu_s, \nu_{s'}) = \sup_{f:\alphabet^\omega \to \{0,1\}} \int fd\nu_{s} - \alpha\int f d\nu_{s'} = V$$
or
$$K_\alpha(\dm)(\nu_s, \nu_{s'}) = \sup_{f:\alphabet^\omega \to \{0,1\}} \int f d\nu_{s'} - \alpha\int f d\nu_{s} = V$$

Assume wlog that we are in the first case and remember by Corollary~\ref{cor:indicator}, $f:\alphabet^\omega \to \{0,1\}$ are indicator functions. Consider such an $f$ with $|V - (\int fd\nu_{s} - \alpha\int f d\nu_{s'})| \le \epsilon/2$ (made possible by supremum).

If we show that $\forall \epsilon \ \exists h$ such that $\int fd\nu_{s} - \alpha\int f d\nu_{s'} - K(\dm^h)(\nu_s,\nu_{s'}) \le \epsilon /2$ then $V - K(\dm^h)(\nu_s,\nu_{s'}) \le \epsilon$ and we are done.

Let $A = \{x \in A^\omega \ | \ f(x) = 1\}$. Since $f$ measurable then $A \in \setofmeasurables$ then:

$$ \int fd\nu_{s} - \alpha\int f d\nu_{s'} = \nu_s(A) - \alpha\nu_{s'}(A)$$

By Lemma~\ref{lem:genset}, we know $\nu_s(A)$ and $\nu_{s'}(A)$ can be approximated by cylinder sets since $\mcal{F} \subseteq S$ as $\mcal{F}$ is the smallest $\sigma$-algebra generated by combinations of cylinders, so contained in $S$ approximable by cylinders. Let $C$ be a cylinder set such that 
$$\nu_s(A) -\sum_{c \in C} \nu_{s}(c) \le \frac{\epsilon}{4\alpha} \text{ and }\nu_{s'}(A) -\sum_{c \in C} \nu_{s'}(c) \le\frac{\epsilon}{4\alpha}$$

We assume without loss of generality that all prefixes in $C$ are of length $h$. If not generalise all prefix's shorter than the maximum length $h$ by adding all possible suffixes to make length $h$.

\begin{align*}
&\sum_{c \in C} \nu_{s}(c) - \alpha  \sum_{c \in C} \nu_{s'}(c) 
\\ & \le \sum_{c \in C} g(c) (\nu_{s}(c) - \alpha \nu_{s'}(c)) 
\\ & \le K_\alpha(\dm^h)(\nu_{s}, \nu_{s'})
\end{align*}
$$g(c) = \begin{cases} 1 & \nu_{s}(c) - \alpha \nu_{s'}(c) > 0 \\ 0 & \text{ otherwise } \end{cases}$$

If $ K_\alpha(\dm^h)(\nu_{s}, \nu_{s'}) > \nu_s(A) - \alpha\nu_{s'}(A)$ then as $K_\alpha(\dm)(\nu_{s}, \nu_{s'}) \ge K_\alpha(\dm^h)(\nu_{s}, \nu_{s'})$ so $K_\alpha(\dm)(\nu_{s}, \nu_{s'}) - K_\alpha(\dm^h)(\nu_{s}, \nu_{s'}) \le \frac{\epsilon}{2} \le \epsilon$ already.

Otherwise
\begin{align*} 
& \nu_s(A) - \alpha\nu_{s'}(A) - K_\alpha(\dm^h)(\nu_{s}, \nu_{s'})
\\ & \le \nu_s(A) - \alpha\nu_{s'}(A) - (\sum_{c \in C} \nu_{s}(c) - \alpha  \sum_{c \in C} \nu_{s'}(c))
\\ &\le |\nu_s(A)- \sum_{c \in C} \nu_{s}(c)| + \alpha|\nu_{s'}(A) -\sum_{c \in C} \nu_{s'}(c)| 
\\ &\le \frac{\epsilon}{4\alpha} + \alpha\frac{\epsilon}{4\alpha} \le \frac{\epsilon}{2}
\end{align*}
\qed
\end{proof}

\subsubsection{Proof of Lemma~\ref{lem:appen:mainbound}}
\begin{proof}
To show $K_\alpha(\dm)(\nu_s,\nu_{s'}) \le \bd{}$ we show by induction on $h$:

\[
K(\dm^h)(\nu_s,\nu_{s'}) \le bd_\alpha(s,s)
\] 
for all $h$ and the result follows from Lemma~\ref{lem:indenough}.

\paragraph{\textbf{Base Case:} $h=0$}
In the base case, we show that $K_\alpha(\dm^0) =0 $. Thus it is necessarily smaller than any value of $bd_\alpha$.

$$K_\alpha(\dm^h)(\nu_s,\nu_{s'})  =  \sup_{\substack{f: A^\omega \to [0,1]\\ \Delta_\alpha(f(t),f(t')) \le \dm^h(t,t) \ \forall t,t' \in A^\omega}}\Delta_\alpha(\hat{f}(\nu_s), \hat{f}(\nu_{s'})) $$
But since $\dm^h = 0$, then 
$$K_\alpha(\dm^h)(\nu_s,\nu_{s'})  =  \sup_{\substack{f: A^\omega \to [0,1]\\ \Delta_\alpha(f(t),f(t')) \le 0 \ \forall t,t' \in A^\omega}}\Delta_\alpha(\hat{f}(\nu_s), \hat{f}(\nu_{s'})) $$
Then we know by Lemma~\ref{lem:expectzero} that $\Delta_\alpha(f(t),f(t')) \le 0$ for all $t,t' \in A^\omega$ then $\Delta_\alpha(\hat{f}(\nu_s), \hat{f}(\nu_{s'})) = 0$ and we obtain:
$$K_\alpha(\dm^h)(\nu_s,\nu_{s'})  =  \sup_{f: A^\omega \to [0,1]\ |\ \Delta_\alpha(f(t),f(t')) \le 0}\{0 \} = 0$$

Therefore $K_\alpha(\dm^h)(\nu_s,\nu_{s'}) =0 \le bd_\alpha(s,s')$

\paragraph{\textbf{Induction Case:}}

We assume $K_\alpha(\dm^h)(\nu_s,\nu_{s'}) \le bd_\alpha(s,s')$ and show $K_\alpha(\dm^{h + 1})(\nu_s,\nu_{s'}) \le bd_\alpha(s,s')$

Case 1:
$\lab(s) \ne \lab(s')$ then $bd_\alpha(s,s') = 1 \ge K_\alpha(\dm^{h+1})(\nu_s,\nu_{s'})$

Case 2: $\lab(s) = \lab(s') = a \in \alphabet$

The strategy is to consider a function $f: \alphabet{}^\omega \to [0,1]$ valid under $K_\alpha(\dm^{h+1})$. We will construct a function $g:S\to [0,1]$ valid under the expansion of $\bd{}$ with the same difference of expectations. This will allow us to conclude the sumpremum in $K_\alpha(\dm^{h+1}) \le \bd{}$

We show a function that is valid under $\dm^{h+1}$ to first find a function valid under $\dm^h$.

Consider an $f$ such that $\Delta_\alpha(f(t),f(t')) \le \dm^{h+1}(t,t')$ for all $t,t' \in A^\omega$ and let $f_a(t) = f(at)$.

\begin{property}\label{prop:faisvalid} $\Delta_\alpha(f_a(t), f_a(t')) \le \dm^{h}(t,t')$.
\end{property}
\begin{align*}
\Delta_\alpha(f_a(t), f_a(t')) &= \Delta_\alpha(f(at), f(at')) \tag{by defn of $f_a$}\\
&\le \dm^{h+1}(at,at') &= \begin{cases} 0 & at^{h} = at'^{h} \\ 1 & \text{otherwise} \end{cases} \tag{by assumption on $f$} \\
&= \dm^h(t,t') &=  \begin{cases} 0 & t^{h} = t'^{h} \\ 1 & \text{otherwise} \end{cases}
\end{align*}

Given this function we find a function $g$ which is valid in the restrictions in $K_\alpha(\bd{}) = \bd{}$ thus the difference of expectations of $g$ is below $\bd{}$.

Let $$g(s) = \hat{f}_a(\nu_s) =\int_{A^\omega} f_a d\nu_{s}$$

\begin{property}\label{prop:gvalid}For all $s,s'$: $\Delta_\alpha(g(s), g(s'))\le bd_\alpha(s,s')$
\end{property}
\begin{align*}
\Delta_\alpha(g(s), g(s')) &= \Delta_\alpha(\hat{f}_a(\nu_s),\hat{f}_a(\nu_s') \tag{by definition of g} \\
&\le K_\alpha(\dm^{h})(\nu_s,\nu_{s'}) 
\\ &= \sup_{\substack{f: A^\omega \to [0,1]\\ \Delta_\alpha(f(t),f(t')) \le \dm^h(t,t') \forall t,t' \in A^\omega}}\{\Delta_\alpha( \hat{f}(\nu_s),  \hat{f}(\nu_{s'}) ) \} \tag{$f_a$ is one such $f$, Property~\ref{prop:faisvalid}} \\ 
&\le bd_\alpha(s,s') \tag{induction assumption}
\end{align*} 

We show the difference of expectation of $g$ is equal to the difference of expectation of $f$:

\begin{property}\label{prop:sameexp}$\hat{f}(\nu_s)= \hat{g}(\mu_s)$
\end{property}
\begin{align*}
\hat{f}(\nu_s) &= \int_{A^\omega} f d\nu_s \\
&= \int_{A^\omega \text{ starting a}} f d\nu_s \tag{s can only yield $a$}\\
&= \int_{A^\omega} f_a \Sigma_{s_i \in S} \mu_s(s_i) d\nu_{s_i} \tag{$s\xrightarrow{a} \mu_s$} \\
&= \Sigma_{s_i \in S}  \mu_s(s_i)  \int_{A^\omega} f_a  d\nu_{s_i}\\ 
&= \Sigma_{s_i \in S}  \mu_s(s_i)  g(s_i) \tag{definition of $g$}\\ 
&= \hat{g}(\mu_s) \tag{expectation of $g$}\\ 
\end{align*}

This allows us to plant the difference of expectation of $g$ (and thus $f$) below the supremum defining the fixed point and thus below the fixed point.

\begin{property}\label{prop:main}
$\Delta_\alpha(\hat{f}(\nu_s), \hat{f}(\nu_{s'})) \le bd_\alpha(s,s') $
\end{property}
\begin{align*}
\Delta_\alpha(\hat{f}(\nu_s), \hat{f}(\nu_{s'})) &= \Delta_\alpha(\hat{g}(\mu_s), \hat{g}(\mu_{s'})) \tag{Property~\ref{prop:sameexp}}\\
&\le K_\alpha(bd_\alpha)(\mu_s, \mu_{s'}) = \sup_{\substack{f: S \to [0,1]\\ \Delta_\alpha(f(s),f(s')) \le bd_\alpha(s,s') \forall s,s' \in S}}\{\Delta_\alpha( \hat{f}(\mu_s),  \hat{f}(\mu_{s'}) ) \}  \tag{g is such an $f$, by Property~\ref{prop:gvalid}} \\ 
&\le bd_\alpha(s,s') \tag{$\Gamma_\alpha(bd_\alpha) \le bd_\alpha$ by definition}
\end{align*}

Since by Property~\ref{prop:main} we have $\Delta_\alpha(\hat{f}(\nu_s), \hat{f}(\nu_{s'})) \le bd_\alpha(s,s') $ for all $f$ satisfying $\dm^{h+1}$ then we also have it in the supremum.

Thus:

\begin{align*}
K_\alpha(\dm^{h + 1})(\nu_s,\nu_{s'}) &= \sup_{\substack{f: A^\omega \to [0,1] \\ \Delta_\alpha(f(t),f(t')) \le \dm^{h+1}(t,t') \forall t,t' \in A^\omega}}\Delta_\alpha( \hat{f}(\nu_s),  \hat{f}(\nu_{s'}) )
\\ & \le bd_\alpha(s,s')
\end{align*}

\end{proof}

\section{Dual form of Skewed Kantorovich (Proof of Lemma~\ref{thm:dualform})}

The primal form can be expressed as the following $LP$:
\[
\begin{array}{l}
\max_{\vec{f} \in [0,1]^n}  \vec{f} \cdot \vec{c} \\
\text{Subject to }\\
A \vec{f} \le \vec{m} \\\end{array}
\]
Where:
\[
\begin{array}{l}
\vec{c}_i = \mu_i -\alpha \upsilon_i \\
A_{(i,j), i} = 1\\
A_{(i,j), j} = -\alpha\\
A_{(i,i), i} = 1-\alpha
\end{array}
\]

We consider an equivalent LP in standard form, to achieve this we separate $\vec{f} \in [0,1]^n$ into two vectors $\vec{a},\vec{b} \in [0,1]^n$, with $\vec{a} = \vec{b}$ (intended to be equal to $\vec{f}$). We obtain the following linear program in standard form:
 
$$\max_{ \vec{a},\vec{b} \in R_+^n}(\mu, -\alpha\nu)\begin{pmatrix}
  \vec{a} \\ \vec{b}
  \end{pmatrix}$$
  Subject to:
  $$ \begin{pmatrix}
  A & A' \\
 I & -I \\
 -I & I \\
 I & 0 
  \end{pmatrix} \begin{pmatrix}
  \vec{a} \\ \vec{b}
  \end{pmatrix} \le  \begin{pmatrix}
   \vec{m} \\
   \vec{0} \\
   \vec{0} \\
   \vec{1}
  \end{pmatrix}
$$
Where \\
\[
\begin{array}{l}
A_{(i,j), i} = 1\\
A'_{(i,j), j} = -\alpha\\
I \text{ is the standard identity matrix } \\
0 \text{ is the zero matrix}
\end{array}
\]

$(A,A') \begin{pmatrix}
  \vec{a} \\ \vec{b}
  \end{pmatrix}\le \vec{m}$ specify the skewed distance must be less than $m$. $(I, - I)\begin{pmatrix}
  \vec{a} \\ \vec{b}
  \end{pmatrix} \le \vec{0}$ and $(-I, I)\begin{pmatrix}
  \vec{a} \\ \vec{b}
  \end{pmatrix} \le \vec{0}$ work to specify $\vec{a} = \vec{b}$. $(I, 0)\begin{pmatrix}
  \vec{a} \\ \vec{b}
  \end{pmatrix} \le \vec{1}$ specifies $\vec{a} \le \mathbf{1}$ (the condition $\vec{f} \in [0,1])$. 

This results in the following dual form following standard techniques.
\begin{align*}
&\min_{\omega \in [0,1]^{n\times n}, \gamma,\tau,\eta\in[0,1]^n} \sum_{i,j} \omega_{i,j} m(i,j) + \sum_i \eta_i\\
& \text{Subject to:} \\
&\forall i: \sum_{j} \omega_{i,j} -\gamma_i + \tau_i  + \eta_i \ge \mu_i \\
&\forall j: \sum_{i} \alpha\omega_{i,j} -\gamma_j + \tau_j \le \alpha\nu_j
\end{align*}

It would be possible to merge $\tau_i, \gamma_i \in [0,1]$ into a single variable in $[-1,1]$, but for consistency with the standard requirement of linear programs and the transportation problem intuition we keep them separate.

We can set the first equation to equality, since there is no sense over-satisfying it (reduce $\omega_{i,j}, \tau_i, \eta_i$ or increase $\gamma_i$ with no violation of other constraints). For parity with standard couplings we would like to remove $\gamma, \tau, \eta$ and $\le$. %

\end{document}